\documentclass[english]{lipics-v2019}
\usepackage[T1]{fontenc}
\usepackage{verbatim}
\usepackage{amsmath}
\usepackage{amsthm}
\usepackage{amssymb}

\makeatletter
\theoremstyle{plain}
\newtheorem{thm}{\protect\theoremname}
\theoremstyle{plain}
\newtheorem{lem}[thm]{\protect\lemmaname}
\theoremstyle{plain}
\newtheorem{cor}[thm]{\protect\corollaryname}
\newtheorem{question}{Question}

\nolinenumbers 
\hideLIPIcs 

\makeatother

\usepackage[svgnames]{xcolor}

\renewcommand{\varepsilon}{\epsilon}

\usepackage[textsize=tiny,textwidth=1.5cm]{todonotes}

\renewcommand{\lg}{\log}

\newcommand{\poly}{\operatorname{poly}}
\newcommand{\friends}{\textsf{neighbors}}

\newcommand{\link}{\textsf{link}}
\newcommand{\cut}{\textsf{cut}}
\newcommand{\markOP}{\textsf{mark}}
\newcommand{\unmarkOP}{\textsf{unmark}}
\newcommand{\nearestMarkedNeighbor}{\textsf{nearestMarkedNeighbor}}
\newcommand{\query}{\textsf{query}}

\newcommand{\D}{\mathcal{D}}

\usepackage{babel}
\providecommand{\corollaryname}{Corollary}
\providecommand{\lemmaname}{Lemma}
\providecommand{\theoremname}{Theorem}

\global\long\def\N{\mathbb{N}}%

\title{Explicit and Implicit Dynamic Coloring of Graphs with Bounded Arboricity}

\author{Monika Henzinger}{Faculty of Computer Science, University of Vienna, Vienna, Austria}{monika.henzinger@univie.ac.at}{https://orcid.org/0000-0002-5008-6530}{The research leading to these results has received funding from the European Research Council under the European Community's Seventh Framework Programme (FP7/2007-2013) / ERC grant agreement No.~340506.}
\author{Stefan Neumann}{Faculty of Computer Science, University of Vienna, Vienna, Austria}{stefan.neumann@univie.ac.at}{}{Stefan Neumann gratefully acknowledges the financial support from the Doctoral Programme ``Vienna Graduate School on Computational Optimization'' which is funded by the Austrian Science Fund (FWF, project no.~W1260-N35). The research leading to these results has received funding from the European Research Council under the European Community's Seventh Framework Programme (FP7/2007-2013) / ERC grant agreement No.~340506.}
\author{Andreas Wiese}{Department of Industrial Engineering, Universidad de Chile, Santiago, Chile}{awiese@dii.uchile.cl}{}
{Andreas Wiese was supported by the grant Fondecyt Regular 1170223.}
\authorrunning{M.~Henzinger, S.~Neumann and A.~Wiese}

\acknowledgements{We are grateful to Sayan Bhattacharya and Shahbaz Khan for
	many helpful discussion about the data structure presented
	in~\cite{bhattacharya15space}.}

\Copyright{Monika Henzinger, Stefan Neumann and Andreas Wiese}

\keywords{dynamic algorithms, vertex coloring, arboricity}

\ccsdesc[500]{Theory of computation~Dynamic graph algorithms}
\ccsdesc[500]{Theory of computation~Data structures design and analysis}

\begin{document}

\maketitle

\begin{abstract}
	Graph coloring is a fundamental problem in computer science. We study the
	fully dynamic version of the problem in which the graph is undergoing edge
	insertions and deletions and we wish to maintain a vertex-coloring with
	small update time after each insertion and deletion.

 	We show how to maintain an $O(\alpha \lg n)$-coloring with polylogarithmic
 	update time, where $n$ is the number of vertices in the graph and $\alpha$
 	is the \emph{current} arboricity of the graph.  This improves upon a result
	by Solomon and Wein~(ESA'18) who maintained an $O(\alpha_{\max}\lg^2 n)$-coloring,
 	where $\alpha_{\max}$ is the maximum arboricity of the graph over all updates.
 
 	Furthermore, motivated by a lower bound by Barba et al.~(Algorithmica'19), we initiate
	the study of \emph{implicit} dynamic colorings.  Barba et al.\ showed that
	dynamic algorithms with polylogarithmic update time
	cannot maintain an $f(\alpha)$-coloring for any function $f$ when the vertex
	colors are stored \emph{explicitly}, i.e.,
 	for each vertex the color is stored explicitly in the memory.
 	Previously, all dynamic algorithms maintained explicit colorings.
	Therefore, we propose to study \emph{implicit} colorings, i.e., the data
	structure only needs to offer an efficient query procedure to return the color of a
	vertex (instead of storing its color explicitly). 
	We provide an algorithm which breaks the lower bound and maintains an
	implicit $2^{O(\alpha)}$-coloring with polylogarithmic update time.
	In particular, this yields the first dynamic $O(1)$-coloring for graphs with
	constant arboricity such as planar graphs or graphs with bounded tree-width,
	which is impossible using explicit colorings.
 
 	To obtain our implicit coloring result we show how to dynamically maintain a
 	partition of the graph's edges into $O(\alpha)$ forests with polylogarithmic
 	update time. We believe this data structure is of independent interest and
 	might have more applications in the future.
\end{abstract}

\thispagestyle{empty}
\newpage{}
\setcounter{page}{1}

\section{Introduction}
	Graph coloring is one of the most fundamental and well-studied problems in
	computer science.  Given a graph $G=(V,E)$ with $n$ vertices, a
	\emph{$C$-coloring} assigns a color from $\{1,\dots,C\}$ to each vertex. The
	coloring is \emph{proper} if all adjacent vertices have different colors.
	The smallest $C$ for which there exists a proper $C$-coloring is called the
	\emph{chromatic number} of $G$.  Unfortunately, it is \textsf{NP}-hard to
	approximate the chromatic number within a factor of $n^{1-\varepsilon}$ for
	all $\varepsilon>0$~\cite{khot06better,zuckerman09linear}. Hence, graph
	coloring is usually studied w.r.t.\ certain graph parameters such as the
	\emph{maximum degree}~$\Delta$ of any vertex or the
	\emph{arboricity}~$\alpha$, which is the minimum number of forests
	into which the edges of $G$ can be partitioned. It is well-known that proper
	$(\Delta+1)$-colorings and proper $O(\alpha)$-colorings can be computed in
	polynomial time.
   
	In the dynamic version of the problem, the graph is undergoing
	edge insertions and deletions and a data structure needs to maintain a
	proper coloring with small update time. More concretely, suppose
	there are $m$ update operations each inserting or deleting a single
	edge. This implies an sequence of graphs $G_0,G_1,\dots,G_m$ such that $G_i$
	and $G_{i+1}$ differ by exactly one edge. Then for each $G_i$ the dynamic
	algorithm must maintain a proper coloring.

	When studying dynamic algorithms w.r.t.\ graph parameters such as the
	maximum degree $\Delta$ or the arboricity $\alpha$ it is important that the
	dynamic algorithms are \emph{adaptive} to the parameter. That is, during a
	sequence of edge insertions and deletions, the values of parameters such as
	$\Delta$ and $\alpha$ might change over time. For example, suppose
	$\alpha(G_i)$ is the arboricity of $G_i$ and let $\alpha_{\max}=\max_i
	\alpha(G_i)$ denote the maximum arboricity of all graphs.
	Then ideally we would like that after the $i$'th update the number of colors
	used by a dynamic algorithm depends on $\alpha(G_i)$ and and not on
	$\alpha_{\max}$ because it might be that $\alpha(G_i) \ll \alpha_{\max}$. 

	Bhattacharya et al.~\cite{bhattacharya18dynamic} studied
	the dynamic coloring problem and showed how to maintain a
	$(\Delta+1)$-coloring with polylogarithmic update time and their algorithm
	is adaptive to the current maximum degree of the graph.  In follow-up
	work~\cite{bhattacharya19fully,henzinger19constant} the update time was
	improved to $O(1)$.

	Later, Solomon and Wein~\cite{solomon18improved} provided a dynamic
	$O(\alpha_{\max} \lg^2 n)$-coloring algorithm with $\poly(\lg\lg n)$ update
	time. Note that the number of colors used by~\cite{solomon18improved}
	depends on maximum arboricity $\alpha_{\max}$ over all graphs $G_i$. Hence,
	we ask the following question.
	\begin{question}
	\label{q:adaptive}
		Are there dynamic coloring algorithms with polylogarithmic update time
		which maintain a coloring that is \emph{adaptive} to the current
		arboricity of the graph?
	\end{question}

	Another interesting question concerns limitations of dynamic coloring
	algorithms.  A lower bound of Barba et al.~\cite{barba19dynamic} shows that
	there exist dynamic graphs which are $2$-colorable but any dynamic algorithm
	maintaining a $c$-coloring must recolor
	$\Omega\left(n^{\frac{2}{c(c-1)}}\right)$ vertices after each update.
	The lower bound holds even for forests, i.e., for graphs with arboricity
	$\alpha=1$.
	This implies that any dynamic algorithm maintaining an $f(\alpha)$-coloring
	for any function $f$ must recolor $n^{\Omega(1)}$ vertices after each
	update. Note that this rules out dynamic
	$O(1)$-colorings for forests and, more generally, planar graphs with
	polylogarithmic update times. 
	
	However, the lower bound only applies to dynamic algorithm that are
	maintaining \emph{explicit colorings}. That is, a coloring is
	\emph{explicit} if after each update the data structure stores an array
	$\mathcal{C}$ of length $n$ such that $\mathcal{C}[u]$ stores the color of
	vertex $u$.  Thus, the color of each vertex can be determined with a single
	memory access.  All of the previously mentioned
	dynamic coloring algorithms maintain explicit colorings but are allowed to
	use more than a constant number of colors.

	In the light of the above lower bound, it is natural to ask whether it can
	be bypassed by \emph{implicit colorings}. That is, a coloring is
	\emph{implicit} if the data structure offers a query routine $\query(v)$ which after
	some computation returns the color of a vertex $v$.
	In particular, we require the following consistency requirement for the
	query operation:
	\begin{itemize}
		\item Consider any sequence of consecutive query operations
		$\query(v_1),\dots,\query(v_k)$ which are not interrupted by an update.
		Then if vertices $v_i$ and $v_j$, $i\neq j$, are adjacent, we have that
		$\query(v_i) \neq \query(v_j)$.
	\end{itemize}
	Note that in the above definition we only consider consecutive query
	operations which are \emph{not interrupted by an update}. This is because
	after an update potentially a lot of vertex colors may change (due to the
	lower bound). Furthermore,
	observe that the definition implies that if we query all vertices of the
	graph consecutively, then we obtain a proper coloring.

	Observe that an explicit coloring always implies an implicit coloring: when
	queried for a vertex $u$, the data structure simply returns
	$\mathcal{C}[u]$. However, implicit colorings are much more versatile than
	explicit colorings: when the colors of many vertices change, this does not
	affect the implicit coloring because it does not have to update the array
	$\mathcal{C}$.  Hence, we ask the following natural question.

	\begin{question}
	\label{q:implicit} Can we break the lower bound of Barba et
	al.~\cite{barba19dynamic} with
		algorithms maintaining implicit colorings?
	\end{question}

	\subsection{Our Contributions}
	We answer both questions affirmatively.

	\textbf{Adaptive explicit colorings.}
	First, we show that there exists a randomized\footnote{
		As usual in the study of randomized dynamic algorithms we assume that the
		adversary is oblivious, i.e., that the sequence of edge insertions and
		deletions is fixed before the algorithm runs.
	}
	algorithm which maintains an explicit and \emph{adaptive} $O(\alpha \lg n)$-coloring with
	polylogarithmic update time.  This answers Question~\ref{q:adaptive}
	affirmatively.
	\begin{thm}
	\label{thm:explicit}
		There is a randomized data structure that maintains an explicit and \emph{adaptive}
		$O(\alpha\log n)$-coloring on a graph with $n$ vertices and arboricity
		$\alpha$ with expected amortized update time $O(\log^2 n)$.
	\end{thm}
	Note that this improves upon the results in~\cite{solomon18improved} in two
	ways: It makes the coloring adaptive and it shaves a $\lg n$-factor in the
	number of colors used by the algorithm.
	To obtain our result, we use a similar approach as the one used
	in~\cite{solomon18improved}.  In~\cite{solomon18improved}, the vertices were
	assigned to $O(\lg n)$ levels and the vertices on each level were
	colored using $O(\alpha_{\max} \lg n)$ colors.  In our result, we assign the
	vertices to $O(\lg^2 n)$ levels and partition the levels into \emph{groups}
	of $O(\lg n)$ consecutive levels each.  We then make sure that for coloring
	the $\ell$'th group we use only $O(2^\ell \lg n)$ colors and that the levels
	of groups with $\ell > \Omega(\lg \alpha)$ are empty.  Then a geometric sum
	argument implies that we use $O(\alpha\lg n)$ colors in total.

	\textbf{Adaptive implicit colorings.}
	Furthermore, we provide two algorithms maintaining implicit colorings. Both
	of these algorithms are also adaptive.
	We first provide an algorithm which maintains an adaptive implicit
	$2^{O(\alpha)}$-coloring with polylogarithmic update time and query time
	$O(\alpha \lg n)$. This improves upon the coloring of
	Theorem~\ref{thm:explicit} for $\alpha=o(\lg\lg n)$.
	\begin{thm}
	\label{thm:implicit1}
		There is a deterministic data structure that maintains an adaptive implicit
		$2^{O(\alpha)}$-coloring, with update time $O(\lg^3 n)$ and query time
		$O(\alpha \lg n)$, where $\alpha$ is the current arboricity of the
		graph.
	\end{thm}

	Note that Theorem~\ref{thm:implicit1} implies that for graphs with constant
	arboricity we can maintain $O(1)$-colorings with polylogarithmic update and
	query times.  This class of graphs contains trees, planar graphs, graphs
	with bounded tree-width, and all minor-free graphs. In particular, this
	breaks the lower bound of Barba et al.~\cite{barba19dynamic} and answers
	Question~\ref{q:implicit} affirmatively.
	\begin{cor}
		There is a deterministic data structure that for dynamic graphs with
		constant arboricity maintains an implicit $O(1)$-coloring, with update
		time $O(\lg^3 n)$ and query time $O(\lg n)$.
	\end{cor}

	Next, we improve upon the results of Theorem~\ref{thm:explicit} and
	Theorem~\ref{thm:implicit1} in the parameter regime
	$\Omega(\lg\lg n) \leq \alpha \leq \lg^{o(1)} n$. More concretely, we obtain the
	following result.
	\begin{thm}
	\label{thm:implicit2}
		There is a deterministic data structure that maintains an adaptive
		implicit $O(\alpha \lg n \cdot \min\{1, \lg\alpha / \lg\lg n\})$-coloring with an
		amortized update time of $O(\lg^2 n)$ and a query time of $O(\lg n)$.
	\end{thm}

	\textbf{Dynamic arboricity decomposition.}
	To derive the results of Theorem~\ref{thm:implicit1}, we introduce a data
	structure which maintains an adaptive \emph{arboricity decomposition} of a
	dynamic graph.  That is, it explicitly maintains a partition of the edges
	of the dynamic graph into $O(\alpha)$~undirected forests. This
	data structure might be of independent interest and might be useful in future
	applications.
	
	To obtain the result we assume that we have black box access to an algorithm
	maintaining a low-outdegree orientation of the graph.  More concretely, a
	\emph{$D$-outdegree edge-orientation} for an undirected graph $G=(V,E)$
	assigns a direction to each edge and ensures that each vertex has outdegree
	at most $D$. We then provide a reduction 
	showing that any data
	structure maintaining a $D$-outdegree orientation of a graph can be turned
	into a data structure for maintaining an arboricity decomposition.  
	\begin{thm}
	\label{thm:arboricity-decomposition}
		Let $G$ be a dynamic graph. 
		Suppose there exists a data structure with (amortized or worst-case
		update) update time $T$ maintaining a $D$-outdegree orientation of $G$.  Then there exists a
		data structure that maintains an arboricity decomposition of $G$ with
		$2D$ forests and with (amortized or worst-case, resp.) update time
		$O(T)$.
	\end{thm}

	Theorem~\ref{thm:arboricity-decomposition} yields the following corollary which we obtain by showing that a
	data structure of Bhattacharya et al.~\cite{bhattacharya15space} can be
	extended to maintain an \emph{adaptive} $O(\alpha)$-outdegree orientation of
	an undirected graph (see Section~\ref{sec:level-data-structure}).
	\begin{cor}
	\label{cor:arboricity-decomposition}
		There exists a deterministic adaptive data structure that maintains a partition of the
		edges into $O(\alpha)$ forests with amortized update time $O(\lg^2 n)$,
		where $\alpha$ is the \emph{current} arboricity of the graph.
	\end{cor}
	The corollary complements a result by Banerjee et al.~\cite{banerjee19fully}
	who presented an algorithm for dynamically maintaining an arboricity
	decomposition consisting of \emph{exactly} $\alpha$ forests with
	$O(m\cdot\poly(\lg n))$ update time, where $m$ is the number of edges
	currently in the graph.  Thus the result in the corollary obtains an
	exponentially faster update time by increasing the number of forests by a
	constant factor.

	See Appendix~\ref{sec:related-work} for more discussions on known data structures for maintaining low outdegree edge-orientations and also further related work.

\section{Level Data Structure}
\label{sec:level-data-structure}
	In this section we introduce a version of the data structure presented
	in~\cite{bhattacharya15space}, which we will refer to as
	\emph{level data structure}.  The data structure dynamically
	maintains an \emph{adaptive} $O(\alpha)$-outdegree orientation of a dynamic graph,
	where $\alpha$ is the current arboricity of the graph.  More precisely, the
	level data structure maintains an undirected graph with $n$ vertices and
	provides an update operation for inserting and deleting edges. It maintains
	an orientation of the edges of the graph such that each vertex has outdegree
	at most $O(\alpha)$. We emphasize that, unlike the data structure presented
	in~\cite{bhattacharya15space}, our level data structure does not require
	that $\alpha$ is an upper bound on the maximum arboricity of the graph over
	the whole sequence of edge insertions and deletions.

	For the rest of the paper, we will write $\{u,v\}$ to denote undirected
	edges and $(u,v)$ to denote directed edges.

	\textbf{Levels, Groups and Invariants.}
	Internally, the data structure maintains a partition of the of the vertices
	into $k=O(\log^2 n)$ levels which we call
	\emph{hierarchy}.  For each $i=1,\dots,k$, we let $V_{i}$ denote the set of
	vertices that are currently assigned to level $i$. Furthermore, we partition
	the levels into \emph{groups} $G_1,\dots,G_{\lceil \lg n\rceil}$
	such that each group contains $L=2+\lceil\log n\rceil$ consecutive levels.
	More precisely, for each $\ell\in\N_{0}$, we set $G_\ell = \left\{ \ell L +
	1,\dots,(\ell+1)L\right\} $. Note that $k=L\cdot\lceil \lg n\rceil$ and that
	neither the total number of level $k$ nor the number of levels per group $L$
	depend on the arboricity.

	The data structure maintains the following invariants for each vertex $v$:
	\begin{enumerate}
		\item\label{inv1}
			If $v \in V_i$, $i<k$ and $i\in G_\ell$, then $v$ has at most
			$5\cdot 2^\ell$ neighbors in $\bigcup_{j\geq i} V_j$.
			That is, each vertex $v$ has at most $5\cdot 2^\ell$
			neighbors at its own or higher levels.
		\item\label{inv2}
			If $v \in V_i$, $i>1$ and $i\in G_\ell$, then 
			$v$ has at least $2^{\ell'}$ neighbors in levels $\bigcup_{j\geq i-1} V_j$,
			where $\ell'$ is such that $i-1\in G_{\ell'}$.
			That is, each vertex $v$ has at least $2^{\ell'}$ neighbors at
			levels $i-1$ and above.
	\end{enumerate}
	Due to edge insertions and deletions the above invariants might get
	violated. If a vertex $v$ does \emph{not} satisfy the invariants, we call it
	\emph{dirty}. Otherwise, we say that $v$ satisfies the
	\emph{degree-property}.

	Note that the above partitioning of the vertices implies an edge
	orientation: For an (undirected) edge $\{u,v\}$ such that $u\in V_{i}$ and
	$v\in V_{i'}$, we assign the orientations as follows: $(u,v)$ if $i<i'$,
	$(v,u)$ if $i>i'$ and an arbitrary orientation if $i=i'$. This corresponds
	to directing an edge from the vertex of lower level towards the vertex of
	higher level in the hierarchy. Note that due to Invariant~\ref{inv1}, each
	vertex at level $i\in G_\ell$ has outdegree at most
	$5\cdot 2^\ell$.

	\textbf{Initialization and Data Structures.}
	The initialization of the data structure is implemented as follows.
	We assume that at the beginning the data structure is given a graph with $n$
	vertices and no edges. We initialize the sets $V_{i}$ by setting $V_{1}:=V$
	and $V_{i}:=\emptyset$ for all $i=2,\dots,k$. The groups $G_\ell$ are
	defined as above and do not depend on the edges of the graph.

	Furthermore, for each vertex $v$ with $v\in V_i$, we maintain the following
	data structures. For each level $i'<i$, we maintain a doubly-linked list
	$\friends(v,i')$ containing all neighbors of $v$ in $V_{i'}$. Furthermore,
	there is a doubly-linked list $\friends(v,\geq i)$ containing all neighbors of $v$ in
	$\bigcup_{j\geq i} V_{j}$. Additionally, for each edge $\{u,v\}$ we store
	a pointer to the position of $v$ in $\friends(u,\cdot)$ and vice versa.
	Note that by additionally maintaining for each list $\friends(v,\cdot)$ the
	number of vertices stored in the list, we can check in time $O(1)$ whether
	one of the invariants is violated for $v$.
	
	\textbf{Updates.}
	Now suppose that an edge $e=\{u,v\}$ is inserted or deleted. Then one of the
	vertices might get dirty and we have to recover 
	the degree-property. While there exists a dirty vertex $v$ with $v\in V_i$,
	we proceed as follows. If $v$ violates
	Invariant~\ref{inv1}, we move $v$ to level $i+1$. If $v$ violates
	Invariant~\ref{inv2}, we move $v$ to level $i-1$.
	Note that during the above process, the algorithm might change the
	levels of vertices $v'$ with $v'\not\in \{u,v\}$. 

	Observe that when a vertex $v$ changes its level due to one of these
	operations, it is easy to update the lists $\friends(v,\cdot)$: When we
	increase the level of $v$ from $i$ to $i+1$, we simply iterate over the list
	$\friends(v,\geq i)$ and split it into lists $\friends(v,i)$ and
	$\friends(v,\geq i+1)$. Furthermore, for each
	$u \in \friends(v, \geq i+1) \cap \bigcup_{j> i} V_{j}$ we have to move $v$ from $\friends(u,i)$ to
	$\friends(u,i+1)$; this can be done in $O(1)$ time\footnote{When
		iterating over $\friends(v,\geq i)$, we use the pointer stored for
		the edge $\{u,v\}$ which provides the position of $v$ in
		$\friends(u,i)$. Now we remove $v$ from $\friends(u,i)$, add $v$ to
		$\friends(u,i+1)$ and update the pointer for edge $\{u,v\}$ accordingly.
	}
	for each such $u$.  Similarly, when we decrease the level of $v$ from $i$ to
	$i-1$, we merge the lists $\friends(v,i-1)$ and $\friends(v,\geq i)$ into a
	single list $\friends(v,\geq i-1)$ and updating the $\friends$-list of all
	vertices in $\bigcup_{j\geq i-1} V_{j}$ similar to the procedure described
	above.

	Observe that when a vertex changes its level via the above routine, we can
	update the edge orientation while iterating over the lists
	$\friends(v,\cdot)$.

	\textbf{Properties.}
	We summarize the properties of the data structure in the following lemma and
	present its proof in Appendix~\ref{sec:proof-level-data-structure}.
	\begin{lem}
	\label{lem:level-data-structure}
		The level data structure is deterministic and has the following
		properties:
		\begin{enumerate}
			\item \label{enu:orientation}
				It maintains an orientation of the edges such that the outdegree
				of each vertex is at most $K\alpha$, where $K=O(1)$ and
				$\alpha$ is the current arboricity of the graph.
			\item \label{enu:update-time}
				Inserting and deleting an edge takes amortized time
				$O(\log^{2}n)$ and each update flips the orientations
				of $O(\lg^2 n)$ edges (amortized).
			\item \label{enu:empty-groups}
				Suppose the graph has arboricity $\alpha$ and set 
				$\ell^* = \lceil \lg \left( 4 \alpha \right) \rceil$.
				Then for all groups $G_\ell$ with $\ell > \ell^*$ and all levels
				$i\in G_{\ell}$, we have that $V_{i}=\emptyset$.
			\item \label{enu:edges-upwards}
				For each (oriented) edge $e=(u,v)$ with $u\in
				V_{i},v\in V_{i'}$ it holds that $i\le i'$.
			\item\label{item:approx-arboricity}
				Returning a value $\alpha^*$ with $\alpha \leq \alpha^* \leq 10 \alpha$ 
				takes time $O(1)$, where $\alpha$ is the current arboricity of the
				graph.
		\end{enumerate}
	\end{lem}

\section{Explicit Coloring with $O(\alpha\log n)$ Colors}
\label{sec:Explicit-coloring}
	We present an algorithm that maintains an $O(\alpha \lg n)$-coloring using
	the level data structure from Section~\ref{sec:level-data-structure}. To
	obtain our coloring, we will assign disjoint color palettes to all levels of
	the data structure. Our main observation is that since the level data
	structure guarantees that each vertex at level $i\in G_\ell$ has at most
	$O(2^\ell)$ neighbors at its own level, it suffices to use
	$O(2^{\ell})$ colors for level $i$. Then a geometric sum
	argument yields that we only use $O(\alpha \lg n)$ colors in total.  As
	before, we do not require an upper bound on $\alpha$ in advance, but the
	number of colors only depends on the \emph{current} arboricity of the graph.

	\textbf{Initialization.} 
	Again, assume that when the data structure is initialized, we are given a
	graph with $n$ vertices and no edges. For this graph, we build the level
	data structure from Section~\ref{sec:level-data-structure}.  Furthermore, to
	each level $i$ in some group $\ell$ we assign a new palette of
	$(K+\varepsilon)\cdot 2^{\ell}$ colors, where $K$ is as in
	Lemma~\ref{lem:level-data-structure} and $\varepsilon = 1/10$.  At the
	very beginning, we assign a random color to each vertex $v\in V$.
	
	Note that the above choice of the color palettes implies that for any
	two levels $i\neq i'$ their color palettes are disjoint.
	
	\textbf{Updates.} 
	Now suppose that an edge $\{u,v\}$ is inserted or deleted. We process this
	update using the update procedure of the level data structure.  Whenever a
	vertex $w$ changes its level in the level data structure, we say that $w$ is
	\emph{affected}. We now provide a re-coloring routine for affected vertices
	and for the vertices $u$ and $v$.
	
	For $u$ and $v$ we proceed as follows. If $u$ and $v$ are in different
	levels, then we do not have to recolor any of them (because the color
	palettes of different levels are disjoint). If $u$ and $v$ are on the same
	level and of different colors, we do nothing. If $u$ and $v$ are on the same
	level $i\in G_\ell$ and have the same color, then suppose that w.l.o.g.\ $u$
	received its current color before $v$ was last recolored.  Now we scan the
	list $\friends(u,\geq i)$ for the colors of all neighbors of $u$ in $V_i$.
	By Lemma~\ref{lem:level-data-structure} there are at most
	$K\cdot 2^{\ell}$ such neighbors and, hence, they use at most
	$K\cdot 2^{\ell}$ different colors. Thus, there must be at least
	$\varepsilon 2^{\ell}$ \emph{available} colors for $u$ in the
	palette of level $i$, i.e., colors that are not used by any of the neighbors
	of $u$ in level $i$. From these available colors, we pick one uniformly at
	random and assign it to $u$. Note that $v$ is not recolored. 

	Whenever an affected vertex $w$ changes its level, we recolor $w$ as
	follows. Suppose that $w$ is moved to level $i \in G_\ell$.  We consider the
	colors of the vertices in $\friends(w,\geq i)\cap V_{i}$ by simply scanning
	the list $\friends(w,\geq i)$.  As before, this yields at least $\varepsilon
	2^\ell$ available colors.  We assign $w$ a random color among
	these available colors.

	\textbf{Analysis.}
	We start by analyzing the update time of algorithm.
	\begin{lem}
		The expected amortized update time of the algorithm is
		$O(\lg^2 n)$.
	\end{lem}
	\begin{proof}
		By Lemma~\ref{lem:level-data-structure}, 
		the amortized update time for the level data structure is 
		$O(\lg^2 n)$.
		Now observe that the work for recoloring affected vertices can be
		charged to the work done by the level data structure:
		When the level data structure moves an affected
		vertex $w$ from level $i$ to a new level $i'\in\{i-1,i+1\}$, then it has to scan all
		neighbors of $w$ in the lists $\friends(w,i)$ and $\friends(w,i')$. When
		the data structure performs these operations, we can keep track of the
		colors of the neighbors of $w$ at the new level $i'$ as described above.
		Thus, the cost for recoloring affected vertices can be charged to the
		running time analysis of the level data structure.
		
		We are left to analyze the recoloring routine for vertices $u$ and $v$
		which are on the same level $i\in G_\ell$. Note that for recoloring $u$,
		the algorithm spends time $O( 2^\ell )$ because the list
		$\friends(u,\geq i)$ has size at most $K \cdot 2^\ell$ by
		Invariant~\ref{inv1}. Now suppose that $u$ is recolored and stays on its
		level $i$. We show that in expectation it takes $\varepsilon 2^\ell$
		edge insertions to vertices on the same level
		until $v$ needs to recolored again: Indeed, suppose that a new edge $\{u,v\}$
		is inserted with $v\in V_i$.  When $v$ received its color, it randomly
		picked one of at least $\varepsilon 2^\ell$ colors and
		thus it picked the same color as $u$ with probability at most
		$\frac{1}{\varepsilon 2^\ell}$.  Now let $X$ be the
		random variable which counts how many such edges from $u$ to vertices on
		the same level as $u$ are inserted until $u$ needs to be recolored.
		Observe that $X$ is geometrically distributed.  Thus, we have that
		$\mathbb{E}[X] = \varepsilon 2^\ell$. This proves the
		claim. By charging $O(1/\varepsilon)$ to each update operation,
		this gives that this recoloring step has an amortized update time of
		$O(1/\varepsilon)$. This running time is subsumed by the update time for
		maintaining the level data structure.
	\end{proof}

	\begin{lem}
	\label{lem:explicit-colors}
		The data structure maintains a $O(\alpha \lg n)$-coloring.
	\end{lem}
	\begin{proof}
		First, recall that for each level $i$ with $i\in G_\ell$ we use
		$(K+\epsilon) 2^\ell$ different colors and that for different
		levels, the color palettes are disjoint. This implies that for each
		group $G_\ell$, we use $(K+\epsilon) 2^\ell \cdot L$ colors.
		Furthermore, by Lemma~\ref{lem:level-data-structure} each level $i$ with
		$i > L \cdot \ell^*$ where
		$\ell^* = \lceil \lg \left( 4 \alpha \right) \rceil$
		satisfies that $V_{i}=\emptyset$.
		Thus, the geometric sum implies that the total number of colors used
		is at most
		\begin{align*}
			\sum_{\ell=0}^{\ell^*} (K+\epsilon) 2^\ell L
			= (K+\epsilon) L \cdot \frac{1-2^{\ell^*+1}}{1 - 2}
			= O(\alpha \lg n).
			&\qedhere
		\end{align*}
	\end{proof}

	The above lemmas imply Theorem~\ref{thm:explicit}.

\section{Dynamic Arboricity Decomposition}
\label{sec:arboricity-decomposition}
	In this section, we present a data structure for maintaining an
	\emph{arboricity decomposition}, i.e., we maintain a partition of the edges
	of a dynamic graph into $O(\alpha)$ edge-disjoint (undirected) forests.
	In particular, we show that any data structure for maintaining an edge
	orientation can be used to maintain such an
	arboricity decomposition, where the number of forests will 
	depend on the maximum outdegree, denoted by $D$ in the sequel.
	We stress that when $D$ depends on some
	parameter (e.g., the arboricity which might increase/decrease after a
	sequence of edge insertions/deletions) then so is the number of forests
	maintained by our data structure.  Using the level data structure from
	Section~\ref{sec:level-data-structure}, this yields that we can maintain an
	arboricity decomposition with $O(\alpha)$ forests if the \emph{current}
	graph has arboricity $\alpha$; the update time is polylogarithmic in $n$.
	We will use this data structure in the next section to give a deterministic
	implicit coloring algorithm.

	For the rest of the section, we assume that we have access to some black box
	data structure that maintains an orientation of the edges with
	update time $T$ for some $T$. We will show how to maintain a set of forests
	$F_0,\dots,F_{2n}$ such that if the maximum outdegree of a node is bounded
	by $D$ (where $D$ which might change over time) the forests
	$F_0,\dots,F_{2D-1}$ provide an arboricity decomposition of the graph and
	the forests $F_{2D},\dots,F_{2n}$ are empty.
	
	\textbf{Initialization and Invariants.}
	We assume that at the beginning we are given a graph with $n$ vertices and
	no edges. For this graph, we build the black box outdegree data structure.
	We initialize $F_0,\dots,F_{2n}$ to $2n$ forests such that each of them contains all vertices $V$
	and no edges. Furthermore, for each vertex $v\in V$ we store an array $A_v$
	storing $n$ bits and initially we set $A_v(i) = 0$ for all $i=0,\dots,n-1$. 
	
	For a vertex $v$, we let $d(v)$ denote the outdegree of $v$ in the black box
	data structure.  When running the data structure, we make sure that the
	following invariants hold for each $v\in V$:
	\begin{enumerate}
		\item\label{inv1-arb} For each $\ell\in\{0,\dots,d(v)-1\}$, either forest
			$F_{2\ell}$ or $F_{2\ell+1}$ but not both contain an out-edge of $v$.
		\item\label{inv2-arb} No out-edge of $v$ is assigned to a forest
			$F_{j}$ with $j\ge 2d(v)$.
		\item\label{inv3-arb} For all $v\in V$ and $\ell\in\{0,\dots,n-1\}$, it
			holds that $A_v(\ell)=1$ iff one of the out-edges of $v$ is assigned to
			forest $F_{2\ell}$ or forest $F_{2\ell+1}$.
	\end{enumerate}
	Observe that when all of the invariants hold, then for each vertex $v$ we
	have that $A_v(\ell)=1$ for $\ell=0,\dots,d(v)-1$ and $A_v(\ell)=0$ for
	$\ell\geq d(v)$. Thus, we have a desired arboricity decomposition. Further
	note that after the initialization of the data structure, all invariants
	hold.

	\textbf{Updates.} Suppose that an edge $\{u,v\}$ is inserted or deleted from
	the graph. We start by inserting or deleting, resp., the edge from the
	black box data structure. Now the black box data structure might either
	(1)~flip the orientation of an existing edge, (2)~add a new out-edge to a
	vertex (due to an edge insertion) or (3)~delete an out-edge of a vertex (due
	to an edge deletion).
	
	Let us start by considering Case~(1), i.e., suppose the black box data
	structure flips the orientation of an edge $\{u',v'\}$. Then we assume
	w.l.o.g.\ that the new orientation is $(u',v')$ and proceed as follows.

	First, we add $(u',v')$ as an out-edge to $u'$.  Let
	$F^*\in\{F_{2d(u')-2},F_{2d(u')-1}\}$ denote the forest in which $v'$ has no
	out-edge (recall that such a forest must exist by Invariant~\ref{inv1-arb}).
	Now we insert the edge $(u',v')$ into $F^*$ and set $A_{u'}(d(u')-1)=1$.  Note
	that after this procedure, all invariants for $u'$ are satisfied.

	Second, we remove the edge $(v',u')$ (with the old orientation) from $v'$.
	Let $\ell^*$ be such that the edge $(v',u')$ was stored in $F_{2\ell^*}$ or
	$F_{2\ell^*+1}$. We remove $(v',u')$ from the corresponding forest and set
	$A_{v'}(\ell^*)=0$. Note that this might violate the invariants because now
	$v'$ has no out-edge in $F_{2\ell^*}$ and $F_{2\ell^*+1}$, but it might have one in
	$F_{2d(v')-2}$ or $F_{2d(v')-1}$ with $d(v')-1 > \ell^*$, where $d(v')$ is the
	outdegree of $v'$ before $(v',u')$ was deleted.  We fix this in the
	next step.

	Third, let $\ell^*$ be as before and set $\ell$ to the largest integer such
	that $A_{v'}(\ell)=1$.  If $\ell < \ell^*$ we do nothing (all invariants
	already hold). Otherwise ($\ell>\ell^*$), we will essentially move the edge
	stored in forest $F_{2\ell}$ or $F_{2\ell+1}$ to forest $F_{2\ell^*}$ or
	$F_{2\ell^*+1}$. More concretely, let $(v',w)$ be the unique out-edge of $v'$
	stored in $F_{2\ell}$ or $F_{2\ell+1}$ and remove $(v',w)$ from this forest.
	Now let $F^*\in\{F_{2\ell^*},F_{2\ell^*+1}\}$ denote the forest in which $w$
	has no out-edge and insert $(v',w)$ into $F^*$. Additionally, set
	$A_{v'}(\ell^*)=1$ and $A_{v'}(\ell)=0$. This restores all invariants for
	$v'$.

	In Case~(2) above, i.e., the black box data structure inserted an out-edge for
	a vertex, we run the first step described above and nothing else. In
	Case~(3), i.e., the black box data structure deleted an out-edge for a
	vertex, we run the second and the third step of the above procedure.

	\textbf{Analysis.}
	First, we show that forests $F_0,\dots,F_{2D-1}$ indeed provide an
	arboricity decomposition of the dynamic graph.
	\begin{lem}
	\label{lem:arb-decomp-properties}
		Let $D$ be the maximum outdegree of any vertex in the outdegree
		decomposition maintained by the black box data structure.
		Then the forests $F_0,\dots,F_{2D-1}$ provide an arboricity
		decomposition of the graph.
	\end{lem}
	\begin{proof}
		Due to Invariant~\ref{inv1-arb}, each edge of the graph is stored in
		some forest. Thus, the union of all forests contains all edges of the
		graph.  Hence, to prove the lemma, it suffices to prove the following
		two claims: (1)~For each $\ell=0,\dots,2D-1$, $F_{\ell}$ does not
		contain a cycle.  (2)~If $\ell\ge 2D$ then $F_{\ell}$ does not contain
		any edges.

		We prove Claim~(1) by contradiction. Suppose that $F_\ell$ contains a
		cycle $C$ over $k$ vertices. Since $C$ is cycle, $C$ contains exactly
		$k$ edges.  By Invariant~\ref{inv1-arb}, each vertex has at most one
		out-edge in $F_\ell$. Hence, $C$ must correspond to a directed cycle in
		$F_\ell$. Now consider the edge $(u',v')$ which closed the cycle when it
		was added to $F_\ell$. In the first step of the algorithm, we only added
		the edge $(u',v')$ to $F_\ell$ if $v'$ had no out-edge in $F_\ell$.
		This contradicts the fact that $(u',v')$ closes a directed cycle.

		Claim~(2) follows directly from Invariant~\ref{inv2-arb} and the
		assumption that $D$ is the maximum outdegree of any vertex.
	\end{proof}

	Note that in the above proof we did not assume that $D$ is an upper bound on
	the maximum outdegree over the \emph{entire} sequence of edge insertions and
	deletions. Instead, we only need that $D$ is the maximum outdegree in the
	\emph{current} graph. Hence, the number of forests providing the arboricity
	decomposition will never be more than $2D$ at any point in time even when
	$D$ is changing over time.

	\begin{lem}
		If the (amortized or worst-case) update time of the black box data
		structure is~$T$, then the (amortized or worst-case, resp.) update time
		of the above algorithm is $O(T)$.
	\end{lem}
	\begin{proof}
		In each update, the algorithm spends time $T$ for inserting or
		deleting, resp., an edge in the black box data structure.
		
		All the other steps can be implemented in $O(1)$ time by maintaining the
		following values:
		(1)~For each $v\in V$, we maintain the maximum index $\ell$ such that
		$A_v(\ell)=1$ (if no such $\ell$ exists we set the corresponding index
		to $-1$). (2)~For each $v$ and $\ell=0,\dots,2n$, we maintain a pointer
		to the copy of $v$ in $F_\ell$.  (3)~For each edge $(u,v)$, we store a
		pointer to the forest $F_\ell$ in which it is currently stored.
		
		Now observe that the first step of the algorithm can be implemented in
		time $O(1)$ as follows: To find $d(u')$, we use the index from~(1). When
		we need to check whether $v'$ has an out-edge in $F_{2d(u')}$ or
		$F_{2d(u')+1}$, we can use the pointers from~(2) to the copies of $v'$
		in $F_{2d(u')}$ and $F_{2d(u')+1}$.

		The second and third step of the algorithm can be implemented similarly.
		When in the second step we have to remove the edge $(v',u')$, we can use
		the pointer from~(3) to find its copy in $O(1)$ time.
	\end{proof}

	The two lemmas above imply Theorem~\ref{thm:arboricity-decomposition}.
	Using the level data structure from Section~\ref{sec:level-data-structure},
	we obtain Corollary~\ref{cor:arboricity-decomposition}.

\section{Implicit Coloring with $2^{O(\alpha)}$ Colors}
\label{sec:Implicit-coloring}
	We present a data structure for implicitly maintaining a
	$2^{O(\alpha)}$-coloring. The data structure has an update time of
	$O(\alpha \lg^3 n)$ and it provides a query operation $\query(u)$ which in time
	$O(\alpha \lg n)$ returns the color of a vertex $u$.
	For planar graphs (which have arboricity at most~$3$) this implies that we
	can maintain an $O(1)$-coloring with update time $O(\lg^3 n)$ and query time
	$O(\lg n)$.

	Our algorithm maintains the arboricity data structure of
	Corollary~\ref{cor:arboricity-decomposition} together with a data structure
	maintaining the forests of the arboricity decomposition. The latter assigns
	a unique root to each tree in the forests.  Our main observation is that for
	any two adjacent vertices $u$ and $v$, there is a tree such that the
	distances of $u$ and $v$ to the root of have \emph{different parity}.  Now
	the query operation for a vertex $u$ picks the color of $u$ based on the
	parities of $u$'s distances to the roots of the trees.

	\textbf{Initialization.}
	We assume that initially the graph has $n$ vertices and no edges.  For this
	graph, we build the arboricity decomposition presented in
	Corollary~\ref{cor:arboricity-decomposition}. Furthermore, each of the
	forests maintained by the arboricity data structure is equipped with the
	data structure from the following lemma.
	\begin{lem}
	\label{lem:top-trees-corollary}
		There exists a data structure for maintaining a
		dynamic forest with the following properties:
		\begin{itemize}
		\item Inserting an edge $\{u,v\}$ into the forest can be done in $O(\lg n)$ time, where
			$u$ and $v$ are in different trees before the edge insertion.
		\item Deleting an edge $\{u,v\}$ from the forest takes $O(\lg n)$ time.
		\item The data structure assigns a unique root to each tree in the
			forest.
		\item For a given vertex $u$, the distance of $u$ to the root of the
			tree containing $u$ can be reported in time $O(\lg n)$.
		\end{itemize}
	\end{lem}
	The lemma is a simple application of dynamic trees or top trees~\cite{alstrup05maintaining}
	and we prove it in Appendix~\ref{sec:proof-top-trees-corollary}.

	\textbf{Updates.}
	Suppose an edge $\{u,v\}$ is inserted or deleted from the graph. Then we
	proceed as follows. First, we insert or delete, resp., the edge in the data
	structure from Corollary~\ref{cor:arboricity-decomposition}. Second, whenever
	the arboricity decomposition inserts or deletes an edge in one of the
	forests, we insert or delete the edge in the corresponding forest of the
	data structure from Lemma~\ref{lem:top-trees-corollary}.

	\textbf{Queries.}
	When we receive a query $\query(u)$ for the color of a vertex $u$, we
	proceed as follows.  For each of the $r = O(\alpha)$ forests we identify the
	tree containing $u$. Let $T_{1},\dots,T_r$ denote these trees. For each
	$T_{j}$, we determine whether the distance of $u$ to the root of $T_{j}$
	using the data structure from Lemma~\ref{lem:top-trees-corollary}. Now for
	each $j=1,\dots,r$, we set $p_{u}(j):=0$ if the distance has even parity
	and $p_{u}(j):=1$ otherwise. Now let
	$p_{u}:=(p_{u}(1),\dots,p_{u}(r))\in\{0,1\}^r$.
	We define the color of $u$ to be $p_{u}$.
	
	\textbf{Analysis.}
	We start by showing that indeed we obtain a $2^{O(\alpha)}$-coloring.

	\begin{lem}
		The data structure maintains an implicit $2^{O(\alpha)}$-coloring.
	\end{lem}
	\begin{proof}
		Consider any edge $\{u,v\}$. We show that query procedure returns
		vectors $p_{u}$ and $p_{v}$ such that $p_u \neq p_v$. Indeed, the edge
		$\{u,v\}$ must be contained in one of the $r=O(\alpha)$ forests
		maintained by the arboricity decomposition.  Therefore, $u$ and $v$ are
		adjacent in some tree $T_{j}$ of that forest. Since $T_j$ has a unique
		root (by Lemma~\ref{lem:top-trees-corollary}), the distances of $u$ and
		$v$ to the root of $T_{j}$ must have a different parity. Thus, we obtain
		that $p_{u}(j)\ne p_{v}(j)$ and, hence, $p_{u}\ne p_{v}$.
		
		Furthermore, the total number of used colors is $2^r = 2^{O(\alpha)}$
		since there are only $2^r$ possibilities for each vector $p_{u}$.
	\end{proof}

	\begin{lem}
		The amortized update time of the data structure is $O(\lg^3 n)$.
		The query time of the algorithm is $O(\alpha \lg n)$.
	\end{lem}
	\begin{proof}
		First, note that the amortized update time of the data structure from
		Corollary~\ref{cor:arboricity-decomposition} is $O(\lg^2 n)$. This
		implies that amortized per update, the arboricity decomposition inserts
		or deletes at most $O(\lg^2 n)$ edges from the forests. For each such
		inserted or deleted edge it takes time $O(\lg n)$ to update the edge in
		the data structure from Lemma~\ref{lem:top-trees-corollary}. This gives
		that the total amortized update time is $O(\lg^3 n)$.

		When answering a query for a vertex $u$, for each of the $O(\alpha)$
		trees containing $u$ we need to query the distance of $u$ to the root
		node of the tree. Each of these queries takes time $O(\lg n)$ by
		Lemma~\ref{lem:top-trees-corollary}. Hence, the total query time is
		$O(\alpha \lg n)$.
	\end{proof}

	The two lemmas above imply Theorem~\ref{thm:implicit1}.

\section{Implicit Coloring with $O(\alpha\lg n\cdot \min\{1, \lg \alpha /\log\log n\})$ Colors}
\label{sec:implicit-levels}
	We present a data structure maintaining an implicit
	$O(\alpha \lg n \cdot \min\{1, \log\alpha/\log\log n \})$-coloring. The
	data structure has an update time of $O(\lg^2 n)$ and a query time of $O(\lg n)$.

	We will now focus on the case that $\alpha \leq \frac{1}{100}\lg n$ and
	provide an algorithm maintaining a $O(\alpha \lg n \log\alpha/\log\log n)$-coloring;
	we will only come back to the case $\alpha > \frac{1}{100}\lg n$ at
	the very end of the section when we prove Theorem~\ref{thm:implicit2}.
	To obtain the result for $\alpha \leq \frac{1}{100}\lg n$, we use the level data structure described in
	Section~\ref{sec:level-data-structure} and an idea similar to that of
	Section~\ref{sec:Explicit-coloring}.  Recall that in
	Section~\ref{sec:Explicit-coloring} we used disjoint color palettes of
	$O(2^\ell)$ colors for each level in group $G_\ell$. Thus, for
	all levels in $G_\ell$ we used $O(2^\ell \cdot \lg n)$ colors
	in total. Now we improve upon this result by providing a query procedure
	which only uses $O(2^\ell \cdot \lg\alpha \lg n / \lg \lg n)$ colors per
	group $G_\ell$.  More concretely, we will partition the group $G_\ell$ into
	$O(\lg\alpha \lg n / \lg \lg n)$ subgroups $S_{\ell,j}$ such that for each subgroup the query
	procedure only uses $O(2^\ell)$ colors.

	\textbf{Subgroups.}
	Recall from Section~\ref{sec:level-data-structure} that the level data
	structure contains $O(\lg^2 n)$ levels and that group $G_\ell$ contains the
	$L$ levels $G_\ell = \left\{ \ell L + 1,\dots,(\ell+1)L\right\}$, where
	$L=2+\lceil\log n\rceil$.
	Now let $\alpha^*$ be an approximation of the arboricity of the graph with
	$\alpha\leq\alpha^*\leq 10\alpha$. We partition each group $G_\ell$ into
	\emph{subgroups $S_{\ell,j}$} of
	$J:=\left\lfloor \log\log n/\log\alpha^*\right\rfloor$ consecutive levels
	each. Formally, for each $\ell\in\N_{0}$ and for each $j\in\{0,\dots,J-1\}$,
	we define that subgroup $S_{\ell,j}$ contains the levels
	$\left\{ \ell L+j\cdot J,\dots,\ell L+(j+1)\cdot J-1\right\}$.  
	Thus, there are $O(\lg n \lg \alpha^* / \lg\lg n)$ subgroups $S_{\ell,j}$ per group
	$G_\ell$.

	Note that $J$ depends on $\alpha^*$ which is an approximation of the
	\emph{current} arboricity of the graph. This implies that as the arboricity
	of the graph changes (due to edge insertions and deletions), the
	subgroups $S_{\ell,j}$ will also change. However, note that the groups
	$G_\ell$ are not affected by this. Also, the algorithm will not need to
	maintain the subgroups $S_{\ell,j}$ explicitly. Instead, it will be
	enough if the algorithm can compute $J$ to check for a given level $i$ in
	which subgroup the level is contained. Later, whenever we need to compute the
	subgroup of a level $i$, we can assume that we know a suitable value for
	$\alpha^*$ and, hence, $J$ with the desired properties via
	Property~\ref{item:approx-arboricity} of
	Lemma~\ref{lem:level-data-structure}.

	Furthermore, to each subgroup $S_{\ell,j}$ we assign a new color palette
	with $(K+\epsilon)2^{\ell}$ colors, where $K$ is the constant
	from Lemma~\ref{lem:level-data-structure} and $\varepsilon = 1/10$. In particular, the palettes for any two
	different subgroups are disjoint.

	\textbf{Initialization.}
	As before, we assume that initially we are given a graph over $n$ vertices
	and without any edges. For this graph we build the level data structure from
	Section~\ref{sec:level-data-structure}.
	We also maintain a counter $t$ which counts the number of edge insertions
	and deletions processed by the data structure, but $t$ does \emph{not} count the
	number of queries processed. Initially, we set $t=0$.  Furthermore, for
	each vertex $v$ we store a pair $(c,t)$ consisting of its color $c$ as well
	as the last time $t$ when its color was last updated. We only store the most recent such
	pair, i.e., when $v$ is assigned a color at time $t$ and there already
	exists a pair $(c',t')$ for $v$ with $t'<t$, we delete the old pair
	$(c',t')$.  At the beginning, we initialize the pairs of all vertices to
	$(0,0)$, indicating that we assigned color 0 after having seen 0 updates. If
	at some time $t$ for a vertex $v$ we store a pair $(c,t')$
	with $t'<t$ then we say that $v$ is \emph{outdated}, otherwise we say that
	$v$ is \emph{fresh}.
	
	\textbf{Updates.} Suppose that an edge is inserted or deleted. Then we
	insert or delete, resp., the edge in the level data structure and update it
	suitably (but do not change the colors stored for the vertices). Also, we
	increase the counter~$t$.
	
	\textbf{Queries.}
	Suppose that the color of a vertex $v$ in a level $i$ is queried at time
	$t$. Then we set $\ell$ and $j$ such that level $i$ is in group $G_\ell$
	and subgroup $S_{\ell,j}$. If $v$ is fresh, we output the color stored for
	$v$. If $v$ is outdated, we recompute the color of $v$ as follows.  First,
	we iterate over $\friends(v,\geq i)$ and find subsets $V'$ and $V''$ which
	are as follows: $V'$ contains the all neighbors of $v$ in some level $i'$
	with $i'>i$ and $i'\in S_{\ell,j}$, $V''$ contains all neighbors of $v$ in
	level $i$.  Second, for each $v'\in V'$ that is outdated, we recursively
	recompute the color for $v'$. Note that after this step all vertices in $V'$
	are fresh.  Third, let $V'''$ be the set of all vertices in $V''$ that are
	fresh.  Observe that $\left|V'\cup V'''\right|\le K \cdot 2^{\ell}$ (by
	Invariant~\ref{inv1} of the level data structure) while the color palette of
	subgroup $S_{\ell,j}$ has $(K+\epsilon)2^{\ell}$ colors.
	Hence, there at least $\varepsilon 2^\ell$ colors which are
	not used by any vertex in $V'\cup V'''$ and we pick one of those and assign
	it to $v$.  Furthermore, we update the pair $(c,t)$ for vertex $v$.
	
	\textbf{Analysis.}
	In the next lemma we show that the coloring assigned to fresh vertices is
	proper. Note that it is enough to prove the claim for fresh vertices: the
	query routine only needs to provide a proper coloring as long as queries are
	not interrupted by an update, thus we only need to consider fresh vertices
	since only fresh vertices are assigned pairs with the most recent timestamp.
	\begin{lem}
	\label{lem:implicit-coloring-proper}
		Let $e=\{u,v\}$ be an edge. Suppose that $u$ and $v$ are fresh and that
		$u$ has color $c_u$ and $v$ has color $c_v$.  Then $c_u \ne c_v$.
	\end{lem}
	\begin{proof}
		Let $i_u$ denote the level of $u$ and let $i_v$ denote the level of $v$.
		Let $(\ell_u,j_u)$ and $(\ell_v,j_v)$ be such that $i_u\in
		S_{\ell_u,j_u}$ and $i_v\in S_{\ell_v,j_v}$. If $u$ and $v$ are from
		different subgroups (i.e., $(\ell_u,j_u)\neq(\ell_v,j_v)$), then we must
		have that $c_u \neq c_v$ since we used disjoint color palettes for
		different subgroups. Now suppose that $u$ and $v$ are from the same
		subgroup (i.e., $(\ell_u,j_u)=(\ell_v,j_v)$).
		We distinguish two cases. First, suppose that $i_u = i_v$. Then
		assume w.l.o.g.\ that $v$ received its color after~$u$. Thus, $u$ was in
		the set $V'''$ when $v$ was colored and, hence, we must have that
		$c_u\neq c_v$.  Second, suppose that $i_u \neq i_v$. W.l.o.g.\ assume
		that $i_u>i_v$. Then $u$ was contained in the set $V'$ when $v$ received
		its color. Hence, we must have that $c_u\neq c_v$.
	\end{proof}

	\begin{lem}
		The algorithm maintains a $O(\alpha \lg \alpha \lg n / \lg \lg n)$-coloring.
	\end{lem}
	\begin{proof}
		We already showed in Lemma~\ref{lem:implicit-coloring-proper} that the
		obtained coloring is proper. It only remains to bound the number of
		colors used by the algorithm.
		First, observe that since each group $G_\ell$ consists of $L=O(\lg n)$
		levels and each subgroup $S_{\ell,j}$ contains $O(\lg\lg n / \lg\alpha)$
		levels, there are $O(\lg n \lg\alpha / \lg\lg n)$ subgroups per group.
		Additionally, for each subgroup $S_{\ell,j}$ we assigned a color palette
		of $(K+\varepsilon)2^\ell$ colors. Thus, the total number
		of colors used per group is 
		$O(2^{\ell} \cdot \lg n \lg\alpha / \lg\lg n)$.
		Now the same geometric sum argument as used in the proof of
		Lemma~\ref{lem:explicit-colors} yields that the data structure uses
		$O(\alpha \lg\alpha \lg n / \lg\lg n)$ colors in total.
	\end{proof}

	Next, we analyze the update and query time of the algorithm. We show that it
	takes $O(\lg n)$ time to query the color of a vertex~$v$. This includes all
	necessary recursive computations for outdated vertices.
	\begin{lem}
		The update time of the algorithm is $O(\lg^2 n)$.
		Furthermore, the query time of the algorithm is $O(\lg n)$.
	\end{lem}
	\begin{proof}
		Since the update procedure of our algorithm only updates the level data
		structure and increases the counter $t$, the update time is $O(\lg^2 n)$
		by Lemma~\ref{lem:level-data-structure}.

		Now let us analyze the query time of the algorithm.
		Consider any subgroup $S_{\ell,j}$ and let $i_{\max}$ be the largest
		level in $S_{\ell,j}$, i.e.,
		$i_{\max}=\ell L+(j+1)\cdot J-1$.
		We prove by induction on the level $i\in S_{\ell,j}$ that for any vertex at level $i$
		the query time is $O(\alpha^{i_{\max}-i+1})$.
		As the algorithm only recolors vertices in $S_{\ell,j}$, we do not have
		to consider any other levels.

		As base case suppose that $i=i_{\max}$. Then we have that $V'=\emptyset$.
		Furthermore, the computation of $V''$ of $V'''$ can be performed in time
		$O(2^\ell) = O(\alpha)$ since $\ell \leq \ell^*$ (see
		Lemma~\ref{lem:level-data-structure}).

		Next, consider a vertex at level $i\in S_{\ell,j}$ with $i<i_{\max}$.
		Then by Invariant~\ref{inv1} of the level data structure, we have that
		the set $V'$ contains at most $O(2^{\ell})=O(\alpha)$
		vertices at levels $i+1,\dots,i_{\max}$. By induction hypothesis,
		coloring each of these vertices takes time
		$O(\alpha^{i_{\max}-(i+1)+1})=O(\alpha^{i_{\max}-i})$. Thus, coloring
		all of these vertices takes time $O(\alpha^{i_{\max}-i+1})$.
		Computing the colors of vertices at level $i$ by computing the sets
		$V''$ and $V'''$ takes time $O(\alpha)$ by the same arguments as in the
		base case. Thus, the total query time for this vertex is
		$O(\alpha^{i_{\max}-i+1})$.

		Now let us bound the total query time. Note that difference of
		$i_{\max}-i$ for $i\in S_{\ell,j}$ is maximized when
		$i=\ell L + j \cdot J$ and in this case $i_{\max}-i=J$. Thus, the total
		query time is at most
		\begin{align*}
			O(\alpha^{i_{\max}-i})
			= O(\alpha^{J})
			\leq O(\alpha^{\lg \lg n / \lg \alpha^*})
			\leq O(\alpha^{\lg \lg n / \lg \alpha})
			= O(2^{\lg \lg n})
			= O(\lg n).
			&\qedhere
		\end{align*}
	\end{proof}

	\begin{proof}[Proof of Theorem~\ref{thm:implicit2}]
		To obtain the data structure claimed in the theorem, we run the above
		algorithm and the explicit algorithm from Theorem~\ref{thm:explicit} in
		parallel.
		After each update, we use Property~\ref{item:approx-arboricity} of
		the level data structure to obtain an approximation $\alpha^*$ of the
		arboricity with $\alpha \leq \alpha^* \leq 10\alpha$. If
		$\alpha^*\leq\frac{1}{10}\lg n$, then we will use the data structure
		from this section for subsequent queries. The previous lemmas imply that
		this provides a $O(\alpha \lg \alpha \lg n / \lg \lg n)$-coloring with
		amortized update time $O(\lg^2 n)$ and query time $O(\lg n)$. If
		$\alpha^*>\frac{1}{10}\lg n$, we use the data structure from
		Theorem~\ref{thm:explicit} for queries. This provides an
		$O(\alpha \lg n)$-coloring with amortized update time $O(\lg^2 n)$ and
		query time $O(1)$ because the coloring maintained by the data structure
		is explicit.
	\end{proof}

\bibliographystyle{abbrv}
\bibliography{dynamic-coloring}

\appendix
\section{Further Related Work}
\label{sec:related-work}
	The first result for dynamic coloring was obtained by Barenboim and
	Maimon~\cite{barenboim17fully} and they showed how to maintain a
	$O(\Delta)$-coloring with worst-case update time $O(\sqrt{\Delta} \poly(\lg n))$. 
	This result was later improved by the
	algorithms~\cite{bhattacharya15space,bhattacharya18dynamic,henzinger19constant}
	to obtain $(\Delta+1)$-colorings with amortized constant update time.
	Duan et al.~\cite{duan19dynamic} provided an algorithm for
	$(1+\varepsilon)\Delta$-edge-coloring with polylogarithmic update time if
	$\Delta\geq\Omega((\lg n/\varepsilon)^2)$.
	Furthermore, algorithms for dynamic coloring were also studied in practice,
	e.g.,~\cite{monical14static,yuan17effective,hardy18tackling}.

	Computing graph colorings of static graphs has been an active research area
	in the distributed community over several decades,
	e.g.,~\cite{linial87distributive,linial92locality,goldberg88parallel,ghaffari17distributed,ghaffari18improved}.
	More recently, Parter et al.~\cite{parter16local} also studied dynamic
	coloring algorithms in the distributed setting.

	Providing dynamic algorithms for graphs with bounded arboricity has been a
	fruitful area of research. Such algorithms have been derived for
	fundamental dynamic problems including shortest
	paths~\cite{frigioni03fully,kowalik06oracles}, maximal independent
	set~\cite{onak18fully},
	matching~\cite{bernstein15fully,bernstein16faster,neiman16simple} or
	coloring~\cite{solomon18improved}.

	Several papers studied the problem of dynamically maintaining low-outdegree
	edge orientation. The first such result was obtained by Brodal and
	Fagerberg~\cite{brodal99dynamic} who obtained an $O(\alpha)$-orientation
	with amortized update time $O(\alpha + \lg n)$.  He et
	al.~\cite{he14orienting} obtained a tradeoff between the outdegree and the
	update time of the algorithm.  Kopelowitz et
	al.~\cite{kopelowitz14orienting} obtained algorithms with worst-case update
	time and this result was improved by Berglin and
	Brodal~\cite{berglin20simple}.  Kaplan and Solomon~\cite{kaplan18dynamic}
	showed how to maintaining edge orientations in the distributed setting when
	the local memory per node is restricted.

\section{Omitted Proofs}
\label{sec:omitted}

\subsection{Proof of Lemma~\ref{lem:level-data-structure}}
\label{sec:proof-level-data-structure}
	Before we prove the lemma, let us first review the data structure by
	Bhattacharya et al.~\cite{bhattacharya15space}. Since the data structure
	of~\cite{bhattacharya15space} was developed for the densest subgraph
	problem, let us first introduce this problem and discuss its relationship with
	arboricity.

	\textbf{Arboricity and densest subgraph.}
	The \emph{density of the densest subgraph} is defined as
	$d^* = \max_{S \subseteq V} \frac{|E(S)|}{|S|}$.
	By the Nash-Williams Theorem, we have that for the arboricity of a graph it
	holds that
	$\alpha = \max_{S \subseteq V} \left\lceil \frac{|E(S)|}{|S|-1} \right\rceil$.
	Thus, we get
	\begin{align}
	\label{eq:arboricity-density}
		\alpha
		= \max_{S \subseteq V} \left\lceil \frac{|E(S)|}{|S|-1} \right\rceil
		\geq \max_{S \subseteq V} \frac{|E(S)|}{|S|-1}
		\geq \max_{S \subseteq V} \frac{|E(S)|}{|S|}
		= d^*.
	\end{align}
	
	\textbf{The data structure of~\cite{bhattacharya15space}.}
	Recall from Section~\ref{sec:level-data-structure} that our data structure
	maintains $\lceil \lg n\rceil \cdot L$ levels. Furthermore, there are
	$\lceil \lg n \rceil$ groups $G_\ell$ consisting of $L$ consecutive levels
	each. Each vertex $v$ at level $i\in G_\ell$ satisfies the following two
	invariants: (1)~$v$ at has at most $5 \cdot 2^\ell$ neighbors in
	$\bigcup_{j\geq i} V_j$ and (2)~$v$ has at least $2^{\ell'}$ neighbors in
	levels $\bigcup_{j\geq i-1} V_j$, where $\ell'$ is such that $i-1\in
	G+{\ell'}$.

	Now the data structure of~\cite{bhattacharya15space} essentially works by
	running $\lceil\lg n\rceil$ data structures $\D_\ell$ in parallel, one data
	structure for each group $G_\ell$. More precisely, each data structure
	$\D_\ell$ stores all vertices and edges of the graph. Furthermore, $D_\ell$
	assigns each vertex
	to exactly one of $L$ levels $U^\ell_1,\dots,U^\ell_L$ and the data structure
	ensures that the following invariants hold: (1)~For each $v\in V$ at level
	$i<L$, $v$ has at most $5 \cdot 2^\ell$ neighbors in $\bigcup_{j\geq i} U^\ell_j$
	and (2)~for each $v\in V$ at level $i>1$, $v$ has at least $2^{\ell}$
	neighbors in levels $\bigcup_{j\geq i-1}U^\ell_j$.  The update procedure of the
	algorithm is the same as described in
	Section~\ref{sec:level-data-structure}, it only takes into account that now
	there are only $L$ levels per data structure $\D_\ell$.
	
	Note that in the data structure of~\cite{bhattacharya15space} there exist
	$\lceil\lg n\rceil$ copies of each vertex $v\in V$, while in our data
	structure each vertex is only stored once.
	
	The data structure by~\cite{bhattacharya15space} satisfies the following
	properties.
	\begin{lem}[{Bhattcharya et al.~\cite{bhattacharya15space}}]
	\label{lem:bhattacharya15space}
		The data structure $\D_\ell$ satisfies the following properties:
		\begin{enumerate}
			\item If $\ell > \lg(4d^*)$, then the highest level of $\D_\ell$ does not
				contain any vertices, i.e., $U^{\ell}_L=\emptyset$.
			\item The amortized update time for maintaining $G_\ell$ is $O(L) = O(\lg n)$.
		\end{enumerate}
	\end{lem}
	The lemma follows from Theorem~2.6 and Theorem~4.2
	in~\cite{bhattacharya15space}.

	Note that since the data structure of~\cite{bhattacharya15space} maintains
	$O(\lg n)$ data structures $\D_\ell$ in parallel, the total update time of
	the data structure becomes $O(\lg^2 n)$.

	\begin{proof}[Proof of Lemma~\ref{lem:level-data-structure}]
		Let us now prove Lemma~\ref{lem:level-data-structure}.

		Property~\ref{enu:edges-upwards} follows immediately from how we defined
		the edge orientation in Section~\ref{sec:level-data-structure}.

		The claim about the update time in Property~\ref{enu:update-time}
		follows from the analysis of the data structures $\D_\ell$
		in~\cite{bhattacharya15space} (the analysis goes through if we assign
		$\Theta(\lceil\lg n\rceil \cdot L)=\Theta(\lg^2 n)$ potential to each edge
		insertion and deletion); the claim for the number of edge flips follows
		from the fact that with amortized update time $O(\lg^2 n)$ the data
		structure cannot flip more than $O(\lg^2 n)$ edges per update
		(amortized).

		Property~\ref{item:approx-arboricity} follows from the fact
		that~\cite{bhattacharya15space} show that a $5$-approximation
		of the densest subgraph can be maintained with amortized update time $O(\lg^2 n)$. Since
		the value of the arboricity and the densest subgraph only differ by a
		factor $2$, we can simply run the data structure of
		\cite{bhattacharya15space} in the background to always have access to
		a value $\alpha^*$ with the desired property.

		Property~\ref{enu:orientation} follows from
		Property~\ref{enu:empty-groups} (which we prove below): By
		Property~\ref{enu:empty-groups} we have that $V_i=\emptyset$ for all
		levels $i$ with $i\in G_\ell$ and $\ell>\ell^*=\lceil\lg(4\alpha)\rceil$.
		Thus, all vertices with out-edges must be in a level $i$ with $i\in
		G_\ell$ with $\ell\leq\lceil\lg(4\alpha)\rceil$. By the invariants
		maintained by the data structure, we obtain that the outdegree of each
		such vertex is at most
		\begin{align*}
			5 \cdot 2^\ell
			= 5 \cdot 2^{\lg(4\alpha)+1}
			= O(\alpha).
		\end{align*}

		We are left to prove Property~\ref{enu:empty-groups} and proceed in two steps.

		First, recall that the levels in the levels data structure from
		Section~\ref{sec:level-data-structure} were denoted
		$V_i$ and those in the data structures $\D_\ell$
		from~\cite{bhattacharya15space} were denoted
		$U^{\ell}_i$. We prove that for all $i\geq 1$ it holds that
		\begin{align}
		\label{eq:claim}
			\bigcup_{j\geq i} V_j 
			\subseteq 
			\bigcup_{j=i'}^{L} U^{\ell}_j,
		\end{align}
		where $\ell \in \{0,\dots,\lceil\lg n\rceil-1\}$ and $i'\in\{1,\dots,L\}$
		are such that $i=\ell \cdot L + i'$.

		Indeed, suppose that $i'=1$. Then we have that $\bigcup_{j=1}^{L} U^{\ell}_j = V$
		since $\D_\ell$ only has $L$ levels and stores all vertices in $V$.
		Thus, the desired subset relationship trivially holds.

		Next, consider $1<i'\leq L$. Observe that by induction hypothesis we have that
		\begin{align*}
			X:= \bigcup_{j\geq i-1} V_j 
			\subseteq 
			\bigcup_{j=i'-1}^{L} U^{\ell}_j =: Y.
		\end{align*}
		This implies for all $v\in V$ it holds that $d_X(v)\leq d_Y(v)$, where
		$d_X(v)$ and $d_Y(v)$ denote the degree of $v$ induced by the vertices
		in $X$ and $Y$, respectively. Since both the level data structure and
		the data structure $\D_\ell$ only promote vertices with at least
		$5\cdot 2^\ell$ vertices to the next level, any vertex which is promoted
		from level $i-1$ to $i$ in the level data structure must also be
		promoted from level $i'-1$ to level $i'$ in the data structure
		$\D_\ell$. Thus, the claim from Equation~\eqref{eq:claim} holds.

		Second, the first property of
		Lemma~\ref{lem:bhattacharya15space} implies that $U^{\ell}_L=\emptyset$
		for $\ell>\lg(4d^*)$.
		By assumption of Property~\ref{enu:empty-groups} and
		Equation~\eqref{eq:arboricity-density} we have that
		\begin{align*}
			\ell
			> \ell^*
			= \lceil \lg(4\alpha) \rceil
			\geq \lg(4 d^*).
		\end{align*}
		Together with the claim above we get that for $i=(\ell^*+1) \cdot L$,
		\begin{align*}
			\bigcup_{j\geq i} V_j
			\subseteq
			U^{\ell^*}_L
			= \emptyset.
		\end{align*}
		Since all levels $i\in G_\ell$ with $\ell > \ell^*$ satisfy
		$i>(\ell^*+1) \cdot L$, we obtain that $V_i=\emptyset$.
		This implies Property~\ref{enu:empty-groups}.
	\end{proof}

\subsection{Proof of Lemma~\ref{lem:top-trees-corollary}}
\label{sec:proof-top-trees-corollary}
	To prove Lemma~\ref{lem:top-trees-corollary} let us first state an application
	of the top trees data structure by Alstrup et al.~\cite{alstrup05maintaining}.
	\begin{lem}[{Alstrup et al.~\cite[Theorem~2.7]{alstrup05maintaining}}]
	\label{lem:top-trees}
		There exists a data structure maintaining a dynamic forest which
		offers the following operations in $O(\lg n)$ time:
		\begin{enumerate}
			\item $\link( \{u,v\} )$, where $u$ and $v$ are in different trees: Insert the
				edge $\{u,v\}$ into the dynamic forest.
			\item $\cut( \{u,v\} )$: Remove the edge $\{u,v\}$ from the dynamic
				forest.
			\item $\markOP(u)$: Mark the vertex $u$.
			\item $\unmarkOP(u)$: Unmark the vertex $u$.
			\item $\nearestMarkedNeighbor(u)$: Return the distance of $u$ to its
			nearest marked neighbor and report the nearest marked neighbor.
		\end{enumerate}
	\end{lem}

	Now we show how to implement the data structure claimed in
	Lemma~\ref{lem:top-trees-corollary}.  For inserting edges we use the
	$\link$-operation and for deleting edges we use the $\cut$-operation from
	Lemma~\ref{lem:top-trees}. We only need to argue how the roots of the trees
	are picked and how to distances to the roots can be computed.
	
	In our construction, we follow the convention that marked vertices (as per
	Lemma~\ref{lem:top-trees}) will correspond to the roots of the trees of
	Lemma~\ref{lem:top-trees-corollary}. We make sure that each tree in the
	dynamic forest contains exactly one marked vertex.

	When we initialize the data structure, we build the top tree data structure
	for a graph with $n$ vertices and without any edges. Furthermore, we mark all
	vertices in the graph using the $\markOP(\cdot)$ operation. Note that
	initially all trees have exactly one root because all connected components
	are isolated vertices and all vertices are marked (and, hence, roots).

	When an edge $\{u,v\}$ is inserted, we proceed as follows. We use the
	$\nearestMarkedNeighbor(u)$ routine to obtain the nearest marked neighbor of
	$u$ and unmark this vertex. Now we use $\link(\{u,v\})$ to link the trees
	of $u$ and $v$ in the top trees data structure. Observe that the resulting
	tree only contains a single marked vertex.

	When an edge $\{u,v\}$ is deleted, we use $\cut(\{u,v\})$ to remove the edge
	$\{u,v\}$ from the top tree data structure. Note that now $u$ and $v$ are
	different trees. Now we query $\nearestMarkedNeighbor(u)$ and
	$\nearestMarkedNeighbor(v)$. Observe that exactly one of them will not have
	a marked neighbor in their new tree.  Suppose this vertex is $u$. Then we
	perform $\markOP(u)$ in the top trees data structure. This implies that each
	tree has a unique root after the operation finished.

	When the data structure is queried for the distance of a vertex $u$ to its
	root, we simply return the distance to $\nearestMarkedNeighbor(u)$.
\end{document}